\documentclass[11pt]{article}
\usepackage[a4paper,margin=1in]{geometry}
\usepackage[T1]{fontenc}
\usepackage[utf8]{inputenc}
\usepackage{lmodern}
\usepackage{amsmath,amssymb,amsthm,mathtools,bm}
\usepackage{array,booktabs,enumitem}
\usepackage[hidelinks]{hyperref}
\hypersetup{
  pdftitle={Evolutionary stability and payoff equivalence in an Eisert-Wilkens-Lewenstein quantum game},
  pdfauthor={Azhar Iqbal and Derek Abbott},
  pdfsubject={Evolutionary stability in an EWL quantum game},
  pdfkeywords={quantum games, EWL scheme, evolutionary stability, entanglement, payoff equivalence}
}
\usepackage{microtype}

\newtheorem{proposition}{Proposition}
\newtheorem{theorem}{Theorem}
\newtheorem{lemma}{Lemma}
\newtheorem{corollary}{Corollary}
\newtheorem{remark}{Remark}
\newtheorem{definition}{Definition}

\newcommand{\tr}{\mathrm{tr}}

\begin{document}

\title{Evolutionary stability and payoff equivalence in an Eisert--Wilkens--Lewenstein quantum game}
\author{Azhar Iqbal\thanks{Corresponding author: \texttt{azhar.iqbal@adelaide.edu.au}}\, and Derek Abbott\\
{\normalsize School of Electrical and Mechanical Engineering,}\\
{\normalsize College of Engineering and Information Technology, Adelaide University,}\\
{\normalsize Adelaide, South Australia 5005, Australia}}
\date{}
\maketitle

\begin{abstract}
An evolutionarily stable strategy (ESS) is a strategy that, once adopted by a
population, cannot be invaded by any sufficiently rare alternative (mutant)
strategy. We analyse evolutionary stability in an Eisert--Wilkens--Lewenstein
(EWL) quantization of a symmetric $2\times2$ coordination game, whose classical
mixed equilibrium $p^*=1/2$ is not an ESS. In the restricted two-parameter EWL
strategy space, the corresponding quantum strategy is a strict symmetric Nash
equilibrium and an ESS for every entanglement strength $\gamma>0$. In the full
$SU(2)$ strategy space it remains a Nash equilibrium, but in the unreduced
game, in which every unitary is a distinct strategy, exact payoff duplicates
preclude strictness and ESS for every pure strategy. For the agreement-payoff
class, the payoff factors through a Bloch map; identifying payoff-equivalent
unitaries gives the reduced space $U(1)\backslash SU(2)\cong S^2$, where
strictness and ESS are restored for $\gamma>0$. An arbitrarily small
perturbation outside this class that preserves $p^*=1/2$ yields a profitable
deviation, so the reduced-space result is not generic.
\end{abstract}

\noindent\textbf{Keywords:} quantum games $\cdot$ Eisert--Wilkens--Lewenstein
scheme $\cdot$ evolutionarily stable strategy $\cdot$ entanglement $\cdot$
payoff equivalence $\cdot$ Nash equilibrium

\section{Introduction}

Quantum game theory studies how quantum states and operations alter the
strategic structure of classical games \cite{EWL,EW,EnkPike,KhanReview}. The
Eisert--Wilkens--Lewenstein (EWL) scheme is a principal formulation for
$2\times2$ games. Its equilibrium predictions, however, depend on the admitted
set of local operations. This dependence was identified early by Benjamin and
Hayden \cite{BenjaminHayden} and has continued to motivate analyses of
restricted and complete strategic domains
\cite{FlitneyHollenberg,Frackiewicz2022,Sakagami2024,Frackiewicz2025}.

The present paper considers evolutionary stability. An evolutionarily stable
strategy (ESS) is a refinement of symmetric Nash equilibrium: a strategy is an
ESS if a population adopting it cannot be invaded by any sufficiently rare
mutant strategy (Definition~\ref{def:ESS})
\cite{MaynardSmithPrice,MaynardSmith,HofbauerSigmund,Lessard,Bomze}. Previous
quantum-game studies have examined evolutionary stability in both EWL and
Marinatto--Weber formulations
\cite{MarinattoWeber,IqbalToorESS,IqbalToorEntanglement,IqbalToorStability,IqbalESS}. Here the
question is examined in the EWL scheme for a symmetric coordination game. The
classical game has an interior symmetric equilibrium $p^*=1/2$ that is not an
ESS, while its two pure coordination equilibria are strict.

\paragraph{Motivation.}
Evolutionary stability is the appropriate equilibrium concept when a game is
played repeatedly by randomly matched members of a large population in which
strategies spread according to their payoffs, rather than once by two fully
rational players \cite{MaynardSmith,HofbauerSigmund}. It is also the concept on
which claims that entanglement changes which strategies can persist in a
population are based \cite{IqbalToorESS,IqbalToorEntanglement,IqbalESS}. Such
claims raise two questions that have not been clearly separated in the
quantum-game literature. The first is the Benjamin--Hayden question: does a
stability result obtained for a restricted set of local unitaries survive when
the full set of local operations is admitted? The second is specific to
evolutionary stability. An ESS must resist every \emph{distinct} mutant, so the
criterion by which two strategies count as distinct is part of the model.
Quantum strategy sets are continua in which physically distinct unitaries can
generate identical measurement statistics, and therefore identical payoffs
\cite{Frackiewicz2025}. The Nash condition depends only on payoff values and is
unaffected by how such strategies are counted. Strictness and ESS, however,
depend on it, and a stability claim that leaves the convention unstated is
consequently incomplete. The coordination game considered here is a simple
setting in which both questions can be answered in closed form and in which the
classical answer, that the mixed equilibrium is not an ESS, is altered by
quantization.

In the restricted two-parameter EWL domain, the quantum strategy
$s^*=(\pi/2,\pi/4)$ associated with $p^*=1/2$ becomes a strict symmetric Nash
equilibrium for every $\gamma>0$ (Theorem~\ref{thm:restricted}). In the full
pure $SU(2)$ domain, the corresponding unitary remains a Nash equilibrium but
is neither strict nor evolutionarily stable (Theorem~\ref{thm:su2}). The
failure is caused by a one-parameter family of distinct unitaries that produces
exactly the same payoffs in this game.

This observation requires a distinction between two normal-form games. In the
\emph{unreduced unitary game}, each labelled element of $SU(2)$ is retained as
a strategy; there, exact payoff duplicates prevent every pure strategy, and
every randomized strategy, from being an ESS. The same conclusion holds after the channel-equivalent pair
$U\sim -U$ is identified, because each payoff fibre contains a continuum of
further, physically distinct unitary channels.
In the \emph{payoff-reduced game}, unitaries that have identical payoffs in
all matchups are identified. This quotient is a modelling choice: it preserves the
payoff function but changes refinements, such as strictness and ESS, that depend
on the identity of strategy types. It does not assert that all unitaries in an
equivalence class are the same physical operation.

For the agreement-payoff class considered below, the EWL payoff depends on a
unitary only through the Bloch vector $n(U)$ defined by
$U^{\dagger}\sigma_zU=n(U)\cdot\boldsymbol{\sigma}$
(Proposition~\ref{prop:reduction}). For fixed $\gamma>0$, the exact
payoff-equivalence classes are the orbits of a left $U(1)$ action, and the
payoff-reduced strategy space is $U(1)\backslash SU(2)\cong S^2$
(Theorem~\ref{thm:fibre}). On this reduced game, the image of $s^*$ is a strict
symmetric Nash equilibrium, and hence an ESS, for every $\gamma>0$; it is not
an ESS at $\gamma=0$ (Theorem~\ref{thm:reduced}). The unreduced and reduced
formulations give different, mathematically consistent answers, since they
rest on different notions of a distinct mutant.

The reduction is not generic. It requires equal rewards for the two agreement
outcomes and equal rewards for the two disagreement outcomes
(Lemma~\ref{lem:class}). Section~\ref{sec:robustness} gives an arbitrarily small
payoff perturbation that preserves the classical equilibrium at $p^*=1/2$ but
makes $s^*$ fail the quantum Nash condition.

\paragraph{Contribution.}
The paper makes four contributions to the study of evolutionary stability in
quantum games. First, it gives a complete closed-form account of how the
evolutionary status of one equilibrium depends on the admitted strategy
domain, extending the Benjamin--Hayden observation from Nash equilibrium to
ESS. Second, it identifies exact payoff equivalence, rather than a change of
payoffs, as the mechanism by which enlarging the strategy space destroys ESS,
and shows that this mechanism removes every pure ESS, and every mixed ESS, from
the unreduced $SU(2)$ game. Third, it characterizes the payoff-equivalence
classes geometrically, as the fibres of the Bloch map
$U\mapsto n(U)$, which explains why the restricted two-parameter result holds
and shows that the two formulations are reconciled by set-valued stability
notions: in the unreduced game the resident is neutrally stable for
$\gamma>0$, and its payoff-equivalence class resists every mutant outside the
class while mutants within it are neutral, a property closely related to an
evolutionarily stable set (Proposition~\ref{prop:ess-set}). Fourth, it determines the scope of these
conclusions: the reduction rests on the absence of single-player terms from the payoff
operator (Section~\ref{sec:when-reduction}), and the equilibrium is lost under
arbitrarily small perturbations that introduce such terms.

The implications for quantum game theory are primarily methodological, but
they bear on applications in which quantum games are used to model populations
of adaptive agents. A statement that entanglement stabilizes a strategy should
specify both the strategic domain and the criterion of strategy identity, and
should state whether the stability holds for a single strategy or for a set of
payoff-equivalent strategies. The operator analysis supplies a concrete way to
test whether the payoff factors through reduced variables; in the present game,
commutation with the entangling operator leads to the Bloch-map reduction.
Section~\ref{sec:when-reduction} then separates this mathematical reduction
from the modelling choice of whether payoff-equivalent strategies should be
identified.

Sections~\ref{sec:preliminaries}--\ref{sec:full-su2} establish the classical,
restricted-domain and full-domain results. Section~\ref{sec:reduction}
develops the payoff quotient, Section~\ref{sec:when-reduction} discusses when
that quotient is appropriate, Section~\ref{sec:robustness} examines games
outside the agreement-payoff class, and Section~\ref{sec:discussion}
summarizes the implications.

\section{Preliminaries on symmetric Nash equilibrium and ESS}

\label{sec:preliminaries} Consider a symmetric two-player game with payoff
function $P(x,y)$, interpreted as the payoff to strategy $x$ against
strategy $y$. A symmetric Nash equilibrium $x^*$ satisfies 
\begin{equation}
P(x^*,x^*)-P(x,x^*)\ge 0 \qquad \text{for all } x.  \label{eq:NE-general}
\end{equation}
It is a \emph{strict} symmetric Nash equilibrium if the inequality in
\eqref{eq:NE-general} is strict for every $x\neq x^*$. The difference
$P(x^*,x^*)-P(x,x^*)$ will be called the Nash gap of $x$ against $x^*$.

\begin{definition}
\label{def:ESS} A symmetric strategy $x^*$ is an evolutionarily stable
strategy (ESS) if for every mutant $x\neq x^*$ either 
\begin{equation}
P(x^*,x^*)>P(x,x^*),  \label{eq:ESS1}
\end{equation}
or 
\begin{equation}
P(x^*,x^*)=P(x,x^*) \quad \text {and} \quad P(x^*,x)>P(x,x).  \label{eq:ESS2}
\end{equation}
\end{definition}

Thus every ESS is a symmetric Nash equilibrium, but not every symmetric Nash
equilibrium is evolutionarily stable. In particular, every strict symmetric
Nash equilibrium is an ESS, because \eqref{eq:ESS1} then holds for every
mutant. We will use \eqref{eq:ESS1} and \eqref{eq:ESS2} in exactly this form
throughout. A strategy satisfying \eqref{eq:NE-general} and, whenever equality
holds there, the weak inequality $P(x^*,x)\ge P(x,x)$ is called
\emph{neutrally stable} \cite{MaynardSmith}; every ESS is neutrally stable.

Two conventions require comment, since both bear directly on the results
below.

\begin{remark}[Strategy domain and mutant populations]
\label{rem:pure} The strategy sets used below are the restricted
two-parameter EWL set, the group $SU(2)$, and the reduced space $S^2$. Each is
a compact metric space, which we equip with its Borel $\sigma$-algebra. A pure
strategy is a point of the set, and a randomized strategy is a Borel
probability measure $\mu$ on it; the pure strategy $x$ is identified with the
Dirac measure $\delta_x$. The payoffs considered here are continuous, hence
bounded and Borel measurable, and are extended bilinearly,
\begin{equation}
P(\mu,\nu)=\int\!\!\int P(x,y)\,d\mu(x)\,d\nu(y).
\label{eq:bilinear-extension}
\end{equation}
This is the expected payoff under random pairwise matching in a large,
well-mixed population. A mutant population is described by a Borel probability
measure $\mu$. It may be read either as a single mutant type that randomizes
according to $\mu$, or as a polymorphic group of pure mutant types distributed
according to $\mu$. Because \eqref{eq:bilinear-extension} is bilinear, both
readings give the same average mutant payoff against any population, and the
ESS conditions are applied to that average. A mutant is \emph{distinct} from
the resident $x^*$ when $\mu\neq\delta_{x^*}$ as measures. No continuity or
smallness of $\mu-\delta_{x^*}$ in strategy space is assumed, so the ESS
property is global over the strategy set. When the resident $x^*$ is
displaced to the population state $(1-\varepsilon)\delta_{x^*}+\varepsilon\mu$
by a mutant fraction $\varepsilon$, the resident is said to resist $\mu$ if
\begin{equation}
P\bigl(x^*,(1-\varepsilon)\delta_{x^*}+\varepsilon\mu\bigr)
>P\bigl(\mu,(1-\varepsilon)\delta_{x^*}+\varepsilon\mu\bigr)
\label{eq:invasion-barrier}
\end{equation}
for all $0<\varepsilon<\bar\varepsilon(\mu)$. For finite games
Definition~\ref{def:ESS} is equivalent to the existence of such an
invasion barrier $\bar\varepsilon(\mu)>0$ for every distinct mutant, and a
barrier can be chosen independently of $\mu$. For infinite strategy sets the
equivalence can fail and a uniform barrier need not exist
\cite{VickersCannings,Bomze}. We therefore state explicitly when a uniform barrier is
obtained (Remark~\ref{rem:mixed}).

The ESS comparisons below are first written against pure mutants, because the
distinction between restricted, unreduced and payoff-reduced strategy spaces is
most transparent at that level. None of the conclusions concerning the stated
resident changes when randomized mutants are admitted. Each negative result
is already witnessed by a pure mutant, and each positive result is extended to
randomized mutants, with an explicit uniform invasion barrier, in
Remark~\ref{rem:mixed}. Remark~\ref{rem:mixed-su2} shows that the unreduced
$SU(2)$ game has no ESS among randomized residents either.

Definition~\ref{def:ESS} is a static, global criterion. For continuous strategy
sets it does not imply stability with respect to small changes in the strategy
adopted by the whole population, which is measured by weak convergence of
measures \cite{OechsslerRiedel}, nor convergence stability under gradual
evolution of the strategy \cite{Eshel}. These dynamic properties are not
claimed here; Remark~\ref{rem:mixed} records the behaviour of the reduced
game in this respect.
\end{remark}

\begin{remark}[Exact payoff duplicates]
\label{rem:duplicates} Distinct strategies $x$ and $x'$ will be called exact
payoff duplicates if
\begin{equation}
P(x,z)=P(x',z)\quad\text{and}\quad P(z,x)=P(z,x')
\qquad\text{for every strategy }z.
\label{eq:duplicate-definition}
\end{equation}
If a candidate $x^*$ has a distinct exact duplicate, neither inequality in
Definition~\ref{def:ESS} can hold against that duplicate. Inserting duplicate
strategy types can destroy strictness and ESS without changing any
payoff values. Conversely, quotienting by exact payoff equivalence preserves
payoffs between equivalence classes but can restore these refinements. The
unreduced game and the quotient game should consequently be distinguished.
Except for the familiar identification $U\sim -U$, payoff-equivalent unitaries
need not represent the same physical operation; their identification is a
game-theoretic reduction, not a claim of physical identity.
\end{remark}

\section{The underlying classical symmetric game}

\label{sec:classical} Consider the symmetric bimatrix game 
\begin{equation}
\begin{array}{c|cc}
& S_1 & S_2 \\ \hline
S_1 & (r,r) & (s,t) \\ 
S_2 & (t,s) & (u,u)%
\end{array}
\label{eq:Matrix}
\end{equation}
and let Alice and Bob play $S_1$ with probabilities $p$ and $q$,
respectively. Alice's expected payoff is 
\begin{equation}
P_A(p,q)=rpq+sp(1-q)+t(1-p)q+u(1-p)(1-q).  \label{eq:Payoffs}
\end{equation}
By symmetry, $P_A(p,q)=P_B(q,p)$, and we write $P(p,q)$ for the payoff to
the $p$-player against the $q$-player.

We specialize to the parameter regime 
\begin{equation}
s=t, \qquad u=r, \qquad r-t>0.  \label{eq:GameDefinition}
\end{equation}
The game then takes the simpler form 
\begin{equation}
\begin{array}{c|cc}
& S_1 & S_2 \\ \hline
S_1 & (r,r) & (t,t) \\ 
S_2 & (t,t) & (r,r)%
\end{array}%
.  \label{eq:Matrix-symmetric-special}
\end{equation}

The pure symmetric equilibria $p=0$ and $p=1$ are strict. Against either
resident pure strategy, every nontrivial mixed deviation assigns positive
probability to the mismatching action and so yields a payoff strictly
below $r$. Both pure equilibria are accordingly ESSs. The analysis below focuses instead on the interior mixed
equilibrium $p^*=1/2$, because it is the classical equilibrium whose
evolutionary status changes in the restricted EWL formulation.

The specialization \eqref{eq:GameDefinition} restricts the game to a
lower-dimensional subfamily of \eqref{eq:Matrix}. It is precisely the class
for which the outcome payoff is an affine function of the total agreement
probability. That property permits the observable reduction developed in
Section~\ref{sec:reduction}.

\paragraph{Outcome probabilities.}
Following the notation customary since the EWL treatment of the Prisoner's
Dilemma, in which $S_1$ and $S_2$ are \emph{cooperate} and \emph{defect}, we
use the labels $C\equiv S_1$ and $D\equiv S_2$ for the two outcomes of each
player. For $X,Y\in\{C,D\}$, let $P_{XY}$ denote the joint probability that
Alice's outcome is $X$ and Bob's outcome is $Y$; the first index always refers
to Alice and the second to Bob. Thus
\begin{equation}
\begin{aligned}
P_{CC}&=\Pr(S_1,S_1), & P_{CD}&=\Pr(S_1,S_2),\\
P_{DC}&=\Pr(S_2,S_1), & P_{DD}&=\Pr(S_2,S_2),
\end{aligned}
\qquad P_{CC}+P_{CD}+P_{DC}+P_{DD}=1,
\label{eq:outcome-probabilities}
\end{equation}
and by \eqref{eq:Matrix} the players' expected payoffs are
$\Pi_A=rP_{CC}+sP_{CD}+tP_{DC}+uP_{DD}$ and
$\Pi_B=rP_{CC}+tP_{CD}+sP_{DC}+uP_{DD}$. In the classical game
$P_{CC}=pq$, $P_{CD}=p(1-q)$, $P_{DC}=(1-p)q$ and $P_{DD}=(1-p)(1-q)$, which
recovers \eqref{eq:Payoffs}. In the quantum game the same four probabilities
are the squared moduli of the amplitudes of the final state in the measurement
basis; see \eqref{eq:P-amplitudes} below. The lemma that follows concerns only
the payoff rule and therefore applies to both.

\begin{lemma}[Characterization of the agreement-payoff class]
\label{lem:class} Let $P_{CC},P_{CD},P_{DC},P_{DD}$ be the outcome
probabilities \eqref{eq:outcome-probabilities}, and let Alice's payoff be
$\Pi_A=rP_{CC}+sP_{CD}+tP_{DC}+uP_{DD}$. The following are equivalent:
\begin{enumerate}[label=(\roman*),nosep]
\item $\Pi_A$ depends on the measurement statistics only through the
agreement probability $P_{CC}+P_{DD}$;
\item $r=u$ and $s=t$.
\end{enumerate}
In that case both players receive the common payoff
\begin{equation}
\Pi_A=\Pi_B=t+(r-t)(P_{CC}+P_{DD}).
\label{eq:agreement-payoff}
\end{equation}
By contrast, the weaker condition $\Pi_A=\Pi_B$ for every probability vector
is equivalent only to $s=t$.
\end{lemma}

\begin{proof}
If $r=u$ and $s=t$, then
$\Pi_A=r(P_{CC}+P_{DD})+t(P_{CD}+P_{DC})$, which gives
\eqref{eq:agreement-payoff} after using $\sum P=1$; the same calculation
applies to Bob's payoff. Conversely, if $\Pi_A$ is constant on every level set
of $P_{CC}+P_{DD}$ in the probability simplex, varying probability between
$CC$ and $DD$ forces $r=u$, while varying it between $CD$ and $DC$ forces
$s=t$. Finally,
\begin{equation}
\Pi_A-\Pi_B=(s-t)(P_{CD}-P_{DC}),
\end{equation}
so equality of the players' payoffs for every probability vector is equivalent
to $s=t$ and imposes no condition on $r$ and $u$.
\end{proof}

Condition \eqref{eq:GameDefinition} accordingly selects those games in which
the players receive the same reward $r$ for either agreement outcome and the
same reward $t<r$ for either disagreement outcome. The condition $s=t$ renders
the game a common-payoff, or doubly symmetric, game, while the further
restriction $u=r$ equates the two coordinated outcomes. The analysis that
follows relies on this stronger agreement-payoff structure, and an explicit
perturbation away from it is examined in Section~\ref{sec:robustness}.

\begin{proposition}
For the game \eqref{eq:Matrix} under \eqref{eq:GameDefinition}, the mixed
strategy 
\begin{equation}
p^*=\frac12
\end{equation}
is a symmetric Nash equilibrium but not an ESS.
\end{proposition}

\begin{proof}
Under \eqref{eq:GameDefinition}, equation \eqref{eq:Payoffs} becomes
\begin{equation}
P(p,q)=r\bigl[pq+(1-p)(1-q)\bigr]
+t\bigl[p(1-q)+(1-p)q\bigr].
\end{equation}
Putting $q=1/2$ gives $P(p,1/2)=(r+t)/2$ for every $p$. Thus $p^*=1/2$ is a
symmetric Nash equilibrium, and the strict inequality \eqref{eq:ESS1} does not
apply against any mutant. The second condition yields
\begin{equation}
P\!\left(\frac12,p\right)-P(p,p)=(r-t)\left(2p(1-p)-\frac12\right).
\label{eq:PayoffDiff1}
\end{equation}
Equivalently,
\begin{equation}
P\!\left(\frac12,p\right)-P(p,p)
=-2(r-t)\left(p-\frac12\right)^2<0
\qquad (p\neq\tfrac12).
\end{equation}
Therefore $p^*=1/2$ is not an ESS.
\end{proof}

At this mixed equilibrium, both players receive 
\begin{equation}
P_A\!\left(\frac12,\frac12\right)=P_B\!\left(\frac12,\frac12\right)=\frac{r+t%
}{2}.  \label{eq:EqPayoffC}
\end{equation}
This value will reappear in the restricted quantum analysis below.

\section{The EWL quantization scheme}

\label{sec:ewl} We use the original EWL matrix convention for the restricted
two-parameter strategic set, 
\begin{equation}
U(\theta,\phi)= 
\begin{pmatrix}
e^{i\phi}\cos(\theta/2) & \sin(\theta/2) \\ 
-\sin(\theta/2) & e^{-i\phi}\cos(\theta/2)%
\end{pmatrix}%
, \qquad \theta\in[0,\pi],\quad \phi\in\left[0,\frac{\pi}{2}\right].
\label{eq:EWLUnitary}
\end{equation}
Let 
\begin{equation}
D=U(\pi,0)= 
\begin{pmatrix}
0 & 1 \\ 
-1 & 0%
\end{pmatrix}%
, \qquad J=\exp\!\left(i\gamma D\otimes D/2\right), \qquad \gamma\in\left[0,%
\frac{\pi}{2}\right].  \label{eq:EWLentangler}
\end{equation}
The initial state is 
\begin{equation}
|\psi_i\rangle=J|S_1S_1^{\prime }\rangle,
\end{equation}
which is separable at $\gamma=0$ and maximally entangled at $\gamma=\pi/2$
\cite{EWL,EWLerratum}. Here $|S_1\rangle,|S_2\rangle$ denote the computational
basis states of Alice's qubit, the primed states $|S_1'\rangle,|S_2'\rangle$
those of Bob's qubit, and $|S_iS_j'\rangle=|S_i\rangle\otimes|S_j'\rangle$.
After the players apply local unitary operations $U_A$ and $U_B$, the
final state becomes 
\begin{equation}
|\psi_f\rangle=J^{\dagger}(U_A\otimes U_B)J|S_1S_1^{\prime }\rangle.
\end{equation}

The generic EWL payoff rule is obtained by weighting the four measurement
outcomes. To avoid confusing the entanglement parameter $\gamma$ with payoff
coefficients, we denote the latter by $a,b,c,d$ in the generic formula: 
\begin{align}
\Pi_A(U_A,U_B)&=a\,|\langle S_1S_1^{\prime }|\psi_f\rangle|^2+b\,|\langle
S_1S_2^{\prime }|\psi_f\rangle|^2+c\,|\langle S_2S_1^{\prime
}|\psi_f\rangle|^2+d\,|\langle S_2S_2^{\prime }|\psi_f\rangle|^2,  \notag \\
\Pi_B(U_A,U_B)&=a\,|\langle S_1S_1^{\prime }|\psi_f\rangle|^2+c\,|\langle
S_1S_2^{\prime }|\psi_f\rangle|^2+b\,|\langle S_2S_1^{\prime
}|\psi_f\rangle|^2+d\,|\langle S_2S_2^{\prime }|\psi_f\rangle|^2.
\label{eq:EWLpayoffs}
\end{align}
For the game \eqref{eq:Matrix} one has $(a,b,c,d)=(r,s,t,u)$, and the four
squared moduli in \eqref{eq:EWLpayoffs} are the outcome probabilities
\eqref{eq:outcome-probabilities}:
\begin{equation}
\begin{aligned}
P_{CC}&=|\langle S_1S_1'|\psi_f\rangle|^2, &
P_{CD}&=|\langle S_1S_2'|\psi_f\rangle|^2,\\
P_{DC}&=|\langle S_2S_1'|\psi_f\rangle|^2, &
P_{DD}&=|\langle S_2S_2'|\psi_f\rangle|^2.
\end{aligned}
\label{eq:P-amplitudes}
\end{equation}
Each $P_{XY}$ is a function of the two local operations and of $\gamma$.

The EWL Prisoner's Dilemma famously exhibits a distinctive equilibrium in
the restricted two-parameter space, but Benjamin and Hayden showed that this
equilibrium need not survive once the strategic space is enlarged to the full
set of deterministic local $SU(2)$ operations \cite{BenjaminHayden}. This strategy-space dependence
is central to the present paper as well.

\section{Evolutionary stability in the restricted two-parameter EWL space}

\label{sec:two-parameter} For the symmetric payoff matrix \eqref{eq:Matrix},
the expected payoffs are 
\begin{equation}
\Pi_A=rP_{CC}+sP_{CD}+tP_{DC}+uP_{DD}, \qquad
\Pi_B=rP_{CC}+tP_{CD}+sP_{DC}+uP_{DD},  \label{eq:payoffs-rstu}
\end{equation}
where $P_{CC},P_{CD},P_{DC},P_{DD}$ are the outcome probabilities
\eqref{eq:P-amplitudes} of the final state, which sum to $1$.

Define 
\begin{equation}
c_A=\cos\frac{\theta_A}{2},\quad s_A=\sin\frac{\theta_A}{2}, \qquad c_B=\cos%
\frac{\theta_B}{2},\quad s_B=\sin\frac{\theta_B}{2}, \qquad
\Phi=\phi_A+\phi_B.  \label{eq:notation-two-parameter}
\end{equation}
In this convention, 
\begin{align}
P_{CC}&=c_A^2c_B^2\bigl(\cos^2\Phi+\cos^2\gamma\,\sin^2\Phi\bigr),
\label{eq:PCC-2p} \\
P_{DD}&=\bigl(s_As_B+\sin\gamma\,\sin\Phi\,c_Ac_B\bigr)^2,  \label{eq:PDD-2p}
\\
P_{CD}&=\bigl(\sin\gamma\,\sin\phi_B\,s_Ac_B-\cos\phi_A\,c_As_B\bigr)^2
+\cos^2\gamma\,\sin^2\phi_A\,c_A^2s_B^2,  \label{eq:PCD-2p} \\
P_{DC}&=\bigl(\sin\gamma\,\sin\phi_A\,c_As_B-\cos\phi_B\,s_Ac_B\bigr)^2
+\cos^2\gamma\,\sin^2\phi_B\,s_A^2c_B^2.  \label{eq:PDC-2p}
\end{align}
A derivation is given in Appendix~\ref{app:derivations}. For the special
symmetric game \eqref{eq:GameDefinition}, both players receive the same
payoff: 
\begin{equation}
\Pi_A=\Pi_B=r(P_{CC}+P_{DD})+t(P_{CD}+P_{DC})=t+(r-t)(P_{CC}+P_{DD}).
\label{eq:payoff-symmetric-2p-raw}
\end{equation}
Combining \eqref{eq:PCC-2p} and \eqref{eq:PDD-2p} yields 
\begin{equation}
P_{CC}+P_{DD}=c_A^2c_B^2+s_A^2s_B^2+2\sin\gamma\,\sin(\phi_A+%
\phi_B)c_Ac_Bs_As_B.  \label{eq:pccpdd-simplified}
\end{equation}
Using 
\begin{equation}
2c_As_A=\sin\theta_A, \qquad 2c_Bs_B=\sin\theta_B, \qquad
c_A^2c_B^2+s_A^2s_B^2=\frac{1+\cos\theta_A\cos\theta_B}{2},
\end{equation}
we obtain the compact common payoff 
\begin{equation}
P(s_A,s_B)=\frac{r+t}{2}+\frac{r-t}{2}\Bigl(\cos\theta_A\cos\theta_B+\sin%
\gamma\,\sin(\phi_A+\phi_B)\sin\theta_A\sin\theta_B\Bigr).
\label{eq:PayoffsQG}
\end{equation}
Equivalently, 
\begin{equation}
P(s_A,s_B)=t+\frac{r-t}{2}\Bigl(1+\cos\theta_A\cos\theta_B+\sin\gamma\,\sin(%
\phi_A+\phi_B)\sin\theta_A\sin\theta_B\Bigr).  \label{eq:PayoffsQG-alt}
\end{equation}

We now write the strategies as $s_A\equiv(\theta_A,\phi_A)$ and $%
s_B\equiv(\theta_B,\phi_B)$. At $\gamma=0$, the classical embedding is
$p=\cos^2(\theta/2)$, while the phase $\phi$ does not affect measurement
probabilities. Hence $\theta=\pi/2$ represents the classical mixed strategy
$p=1/2$, and $s^*=(\pi/2,\pi/4)$ is a convenient representative of that
zero-entanglement equivalence class.

\begin{theorem}
\label{thm:restricted}
Let the equilibrium candidate in the restricted two-parameter EWL strategic space be 
\begin{equation}
s^*=\left(\frac{\pi}{2},\frac{\pi}{4}\right).  \label{eq:sstar-two-parameter}
\end{equation}
Then:

\begin{enumerate}

\item at $\gamma=0$, the strategy $s^*$ is a symmetric Nash equilibrium but
not an ESS;

\item for every $\gamma>0$, the strategy $s^*$ is a strict symmetric Nash
equilibrium, and an ESS, within the restricted two-parameter EWL
space.
\end{enumerate}
\end{theorem}

\begin{proof}
For a general mutant $s=(\theta,\phi)$, equation \eqref{eq:PayoffsQG} gives
\begin{equation}
P(s^*,s^*)-P(s,s^*)=
\frac{r-t}{2}\Bigl[
\cos\theta^*(\cos\theta^*-\cos\theta)
+\sin\gamma\,\sin\theta^*\bigl(\sin\theta^*\sin(2\phi^*)-\sin\theta\sin(\phi+\phi^*)\bigr)
\Bigr].
\label{eq:SymmNE}
\end{equation}
At $s^*=(\pi/2,\pi/4)$ this reduces to
\begin{equation}
P(s^*,s^*)-P(s,s^*)=
\frac{r-t}{2}\sin\gamma\left[1-\sin\theta\,\sin\left(\phi+\frac{\pi}{4}\right)\right].
\label{eq:NEQG}
\end{equation}
If $\gamma=0$, the right-hand side vanishes identically, so $s^*$ is a non-strict symmetric Nash equilibrium. Since
\begin{equation}
P(s^*,s^*)=\frac{r+t}{2},
\end{equation}
this coincides with the classical equilibrium payoff \eqref{eq:EqPayoffC}. To test evolutionary stability at $\gamma=0$, use the second ESS condition. Direct substitution into \eqref{eq:PayoffsQG} gives
\begin{equation}
P(s^*,s)-P(s,s)=
-\frac{1}{2}(r-t)\cos^2\theta
+\frac{1}{2}(r-t)\sin\gamma\,\sin\theta\left[\sin\left(\phi+\frac{\pi}{4}\right)-\sin\theta\sin(2\phi)\right].
\label{eq:ESSSecondQ}
\end{equation}
At $\gamma=0$ this becomes
\begin{equation}
P(s^*,s)-P(s,s)= -\frac{1}{2}(r-t)\cos^2\theta.
\end{equation}
For example, the distinct mutant $s=(\pi/2,0)$ makes this second ESS gap equal
to zero, while mutants with $\cos\theta\neq0$ make it negative. Thus the second
ESS condition fails and $s^*$ is not an ESS at $\gamma=0$.

Now suppose $\gamma>0$. In the EWL domain $\theta\in[0,\pi]$ and $\phi\in[0,\pi/2]$, hence
\begin{equation}
0\le \sin\theta\le 1,
\qquad
\frac{1}{\sqrt{2}}\le
\sin\left(\phi+\frac{\pi}{4}\right)\le 1.
\end{equation}
Their product equals $1$ only at $\theta=\pi/2$ and $\phi=\pi/4$, i.e. only at $s=s^*$. Therefore the bracket in \eqref{eq:NEQG} is strictly positive for every mutant $s\neq s^*$, and because $r-t>0$ and $\sin\gamma>0$, one has
\begin{equation}
P(s^*,s^*)>P(s,s^*)
\qquad \text{for all } s\neq s^*.
\end{equation}
Hence $s^*$ is a strict symmetric Nash equilibrium, and an ESS.
\end{proof}

\begin{remark}
The theorem does not imply that entanglement generically makes
symmetric equilibria evolutionarily stable in EWL games. The statement is
specific to the game \eqref{eq:Matrix} under \eqref{eq:GameDefinition} and
to the restricted two-parameter EWL strategic domain.
\end{remark}

\section{Evolutionary stability in the full pure 
\texorpdfstring{$SU(2)$}{SU(2)} strategy space}

\label{sec:full-su2} We now enlarge the strategy set to the full pure $SU(2)$
family while keeping a convention compatible with \eqref{eq:EWLUnitary}: 
\begin{equation}
U(\theta,\phi,\alpha)= 
\begin{pmatrix}
e^{i\phi}\cos(\theta/2) & e^{i\alpha}\sin(\theta/2) \\ 
-e^{-i\alpha}\sin(\theta/2) & e^{-i\phi}\cos(\theta/2)%
\end{pmatrix}%
, \qquad \theta\in[0,\pi],\quad \phi,\alpha\in[0,2\pi).
\label{eq:qstrategy_su2}
\end{equation}
Setting $\alpha=0$ recovers the matrix \eqref{eq:EWLUnitary}, although the
restricted domain also imposes $\phi\in[0,\pi/2]$. Equation
\eqref{eq:qstrategy_su2} covers $SU(2)$, with the usual coordinate
non-uniqueness at $\theta=0$ and $\theta=\pi$. The matrices $U$ and $-U$
induce the same unitary channel on a qubit and may therefore be identified at
the physical level. This optional global-sign identification does not affect
the results below: the neutral fibres contain a continuum of additional
members that are distinct unitary channels even though the present game
assigns them identical payoffs. We retain the entangler \eqref{eq:EWLentangler}. Let the equilibrium
candidate be 
\begin{equation}
s^*=\left(\frac{\pi}{2},\frac{\pi}{4},0\right),
\label{eq:sstar-three-parameter}
\end{equation}
and let the mutant be $s=(\theta,\phi,\alpha)$.

Define 
\begin{equation}
c_i=\cos\frac{\theta_i}{2}, \qquad d_i=\sin\frac{\theta_i}{2}, \qquad i=A,B.
\end{equation}
Under \eqref{eq:GameDefinition}, both players again receive the same payoff, 
\begin{equation}
P(s_A,s_B)=t+(r-t)(P_{CC}+P_{DD}),  \label{eq:payoff_symmetric_su2}
\end{equation}
where now 
\begin{align}
P_{CC}&=\Bigl(c_Ac_B\cos(\phi_A+\phi_B)-\sin\gamma\,d_Ad_B\sin(\alpha_A+%
\alpha_B)\Bigr)^2 +\cos^2\gamma\,c_A^2c_B^2\sin^2(\phi_A+\phi_B),
\label{eq:PCC-su2} \\
P_{DD}&=\Bigl(\sin\gamma\,c_Ac_B\sin(\phi_A+\phi_B)+d_Ad_B\cos(\alpha_A+%
\alpha_B)\Bigr)^2 +\cos^2\gamma\,d_A^2d_B^2\sin^2(\alpha_A+\alpha_B).
\label{eq:PDD-su2}
\end{align}
Appendix~\ref{app:su2-derivation} shows how these formulas arise from the EWL
circuit.

\begin{theorem}
\label{thm:su2}
In the full pure $SU(2)$ strategic space \eqref{eq:qstrategy_su2}, the
equilibrium
\begin{equation}
s^*=\left(\frac{\pi}{2},\frac{\pi}{4},0\right)
\end{equation}
is a symmetric Nash equilibrium for every $\gamma\in[0,\pi/2]$, but it is
non-strict and is not an ESS for any such $\gamma$.
\end{theorem}

\begin{proof}
Substituting $s^*=(\pi/2,\pi/4,0)$ and a general mutant $s=(\theta,\phi,\alpha)$ into \eqref{eq:payoff_symmetric_su2} gives
\begin{equation}
P(s^*,s)=t+\frac{r-t}{2}\left[1+\sin\gamma\,\sin\theta\,\sin\left(\phi+\frac{\pi}{4}-\alpha\right)\right],
\label{eq:Psstar_s_su2}
\end{equation}
whereas mutant self-play yields
\begin{equation}
P(s,s)=t+\frac{r-t}{2}\left[1+\cos^2\theta+\sin\gamma\,\sin^2\theta\,\sin\bigl(2(\phi-\alpha)\bigr)\right].
\label{eq:Pss_su2}
\end{equation}
Hence the full-$SU(2)$ analogue of \eqref{eq:ESSSecondQ} is
\begin{equation}
P(s^*,s)-P(s,s)=\frac{r-t}{2}\left[-\cos^2\theta+\sin\gamma\,\sin\theta\left(\sin\left(\phi+\frac{\pi}{4}-\alpha\right)-\sin\theta\,\sin\bigl(2(\phi-\alpha)\bigr)\right)\right].
\label{eq:ESSSecondQ_SU2}
\end{equation}
This reduces to the two-parameter expression \eqref{eq:ESSSecondQ} when $\alpha=0$.

Similarly,
\begin{equation}
P(s^*,s^*)-P(s,s^*)=\frac{r-t}{2}\sin\gamma\left[1-\sin\theta\,\sin\left(\phi+\frac{\pi}{4}-\alpha\right)\right].
\label{eq:NEdiff_su2}
\end{equation}
Because $0\leq\sin\theta\leq 1$ and $-1\leq\sin(\phi+\pi/4-\alpha)\leq 1$,
equation \eqref{eq:NEdiff_su2} is non-negative for all mutants. Thus $s^*$
remains a symmetric Nash equilibrium. However, it is not strict. To see this,
consider the family of mutants satisfying
\begin{equation}
\theta=\frac{\pi}{2},
\qquad
\alpha=\phi-\frac{\pi}{4}\pmod{2\pi}.
\label{eq:flat-direction}
\end{equation}
Apart from $\phi=\pi/4$, which gives $U=U^*$, and
$\phi=5\pi/4$, which gives the operationally equivalent unitary $U=-U^*$,
the family consists of physically distinct mutants. Choose any other member,
so that $s\neq s^*$ even after quotienting by the irrelevant global sign.
For these mutants,
\begin{equation}
\sin\theta\,\sin\left(\phi+\frac{\pi}{4}-\alpha\right)=1,
\end{equation}
so \eqref{eq:NEdiff_su2} gives
\begin{equation}
P(s^*,s^*)-P(s,s^*)=0.
\end{equation}
Hence $s^*$ is not a strict symmetric Nash equilibrium for any
$\gamma\in[0,\pi/2]$.

It remains to test the second ESS condition. For the same mutant family \eqref{eq:flat-direction}, one has $\cos\theta=0$ and
\begin{equation}
\sin\bigl(2(\phi-\alpha)\bigr)=\sin\left(2\cdot\frac{\pi}{4}\right)=1.
\end{equation}
Therefore \eqref{eq:ESSSecondQ_SU2} reduces to
\begin{equation}
P(s^*,s)-P(s,s)=\frac{r-t}{2}\sin\gamma(1-1)=0.
\end{equation}
So the second ESS condition fails as well. Consequently $s^*$ is not an ESS in the full pure $SU(2)$ strategy space
for any $\gamma\in[0,\pi/2]$.
\end{proof}

\begin{remark}
\label{rem:no-implication} The neutral family in
Theorem~\ref{thm:su2} is not peculiar to $s^*$. Section~\ref{sec:reduction}
shows that every pure $SU(2)$ strategy has a corresponding family of exact
payoff duplicates for the agreement-payoff game. Corollary~\ref{cor:no-ess-su2}
extends the non-ESS conclusion to the entire unreduced unitary
strategy set. Theorem~\ref{thm:reduced} addresses the distinct question obtained
after quotienting those duplicates.
\end{remark}

\section{The agreement observable and the reduced strategy space}

\label{sec:reduction} The neutral family \eqref{eq:flat-direction} was
found by direct substitution. Its origin is visible in
\eqref{eq:Psstar_s_su2} and \eqref{eq:Pss_su2}: the phases $\phi$ and $\alpha$
enter only through $\phi-\alpha$, and the neutral family is the level set
$\phi-\alpha=\pi/4$. Unitaries along such a level set are not generally the
same physical operation. They are, however, exact payoff duplicates in the
normal-form game considered here. The following proposition identifies the
observable responsible for this payoff equivalence.

\begin{proposition}[Reduction to the agreement observable]
\label{prop:reduction} Let the game satisfy \eqref{eq:GameDefinition} and let
$|\psi_f\rangle=J^{\dagger}(U_A\otimes U_B)J|S_1S_1'\rangle$ with $J$ as in
\eqref{eq:EWLentangler}. Write $\sigma_x,\sigma_y,\sigma_z$ for the Pauli
matrices and put
\begin{equation}
U^{\dagger}\sigma_z U=n(U)\cdot\boldsymbol{\sigma},
\qquad n(U)\in S^2\subset\mathbb{R}^3,
\label{eq:bloch-map}
\end{equation}
where $\boldsymbol{\sigma}=(\sigma_x,\sigma_y,\sigma_z)$. Then
\begin{equation}
P(U_A,U_B)=t+\frac{r-t}{2}\Bigl[1+n(U_A)^{\!\top} M_\gamma\, n(U_B)\Bigr],
\qquad
M_\gamma=\begin{pmatrix} 0 & \sin\gamma & 0\\ \sin\gamma & 0 & 0\\ 0&0&1
\end{pmatrix}.
\label{eq:bloch-payoff}
\end{equation}
In particular the payoff depends on $U_A$ and $U_B$ only through $n(U_A)$ and
$n(U_B)$.
\end{proposition}

\begin{proof}
By Lemma~\ref{lem:class}, $P=t+(r-t)(P_{CC}+P_{DD})$. The agreement
projector is $\Pi_{\mathrm{agree}}=|S_1S_1'\rangle\langle S_1S_1'|+|S_2S_2'
\rangle\langle S_2S_2'|=\tfrac12(\mathbb{I}+\sigma_z\otimes\sigma_z)$, so
$P_{CC}+P_{DD}=\tfrac12\bigl(1+\langle\psi_f|\sigma_z\otimes\sigma_z|\psi_f
\rangle\bigr)$. Now $D=i\sigma_y$, hence $D\otimes D=-\sigma_y\otimes\sigma_y$
and $J=\exp(-i\gamma\,\sigma_y\otimes\sigma_y/2)$. Since $\sigma_z\otimes
\sigma_z$ commutes with $\sigma_y\otimes\sigma_y$, it commutes with $J$ and
with $J^{\dagger}$. Therefore
\begin{align}
\langle\psi_f|\sigma_z\otimes\sigma_z|\psi_f\rangle
&=\langle\psi_i|(U_A\otimes U_B)^{\dagger}J\,(\sigma_z\otimes\sigma_z)\,
J^{\dagger}(U_A\otimes U_B)|\psi_i\rangle \notag\\
&=\langle\psi_i|\bigl(U_A^{\dagger}\sigma_z U_A\bigr)\otimes
\bigl(U_B^{\dagger}\sigma_z U_B\bigr)|\psi_i\rangle
=\langle\psi_i|\,\bigl(n(U_A)\!\cdot\!\boldsymbol{\sigma}\bigr)
\otimes\bigl(n(U_B)\!\cdot\!\boldsymbol{\sigma}\bigr)\,|\psi_i\rangle ,
\end{align}
where $|\psi_i\rangle=J|S_1S_1'\rangle=\cos(\gamma/2)|S_1S_1'\rangle
+i\sin(\gamma/2)|S_2S_2'\rangle$. Evaluating the correlation function of that
state gives $\langle a\!\cdot\!\sigma\otimes b\!\cdot\!\sigma\rangle
=a_zb_z+\sin\gamma\,(a_xb_y+a_yb_x)$, which is \eqref{eq:bloch-payoff}. The
computation is given in Appendix~\ref{app:bloch}.
\end{proof}

Writing $U(\theta,\phi,\alpha)$ as in \eqref{eq:qstrategy_su2}, an elementary
computation (Appendix~\ref{app:bloch}) gives
\begin{equation}
n(\theta,\phi,\alpha)=\bigl(\sin\theta\cos\chi,\;\sin\theta\sin\chi,\;
\cos\theta\bigr), \qquad \chi=\phi-\alpha ,
\label{eq:bloch-coords}
\end{equation}
so that $(\theta,\chi)$ are ordinary spherical coordinates on $S^2$. The map
$U\mapsto n(U)$ has one-dimensional fibres. Distinct unitaries in a
fibre are assigned the same payoffs by \eqref{eq:bloch-payoff}, although they
need not be physically identical. Substituting \eqref{eq:bloch-coords} into \eqref{eq:bloch-payoff} reproduces
\eqref{eq:PayoffsQG}, \eqref{eq:Psstar_s_su2} and \eqref{eq:Pss_su2}, which
provides an independent check of the algebra of
Sections~\ref{sec:two-parameter} and~\ref{sec:full-su2}.

\begin{theorem}[Structural payoff equivalence]
\label{thm:fibre} Let the game satisfy \eqref{eq:GameDefinition}. For any
fixed $\gamma\in(0,\pi/2]$, two pure strategies $U,U'\in SU(2)$ are exact
payoff duplicates in the sense of \eqref{eq:duplicate-definition} if and only
if
\begin{equation}
U'=V_{\psi}U, \qquad
V_{\psi}=\begin{pmatrix}e^{i\psi}&0\\0&e^{-i\psi}\end{pmatrix},
\qquad \psi\in[0,2\pi).
\label{eq:left-orbit}
\end{equation}
The same condition characterizes exact payoff duplication simultaneously for
every $\gamma\in[0,\pi/2]$. At $\gamma=0$ alone, the exact classes are larger:
two strategies are payoff duplicates if and only if $n_z(U)=n_z(U')$. Thus
the equivalence relation common to all positive entanglement strengths, and
still valid at zero entanglement, has quotient
\begin{equation}
U(1)\backslash SU(2)\cong S^2.
\label{eq:left-quotient}
\end{equation}
If the game were reduced separately at $\gamma=0$, the larger classes would
instead be indexed by $n_z\in[-1,1]$. The quotient in
\eqref{eq:left-quotient} is diffeomorphic, by inversion, to the more
conventionally written right quotient $SU(2)/U(1)$.
\end{theorem}

\begin{proof}
By Proposition~\ref{prop:reduction}, for fixed $\gamma$ the payoff difference
against an opponent with Bloch vector $b$ is proportional to
\begin{equation}
\bigl(n(U)-n(U')\bigr)^{\top}M_\gamma b.
\end{equation}
When $\gamma>0$,
$\det M_\gamma=-\sin^2\gamma\neq0$. Since $b$ ranges over all of $S^2$, the
payoff difference vanishes for every opponent if and only if
$n(U)=n(U')$. The payoff in \eqref{eq:bloch-payoff} is symmetric in its two
arguments, so the same condition also gives equality when $U$ and $U'$ occur
in the opponent position. When $\gamma=0$,
$M_0=\operatorname{diag}(0,0,1)$, and exact duplication reduces to
$n_z(U)=n_z(U')$. Equality for every entanglement value has the same exact
classes as any one fixed value $\gamma>0$.

For the group-theoretic description, write $U'=WU$. The equality
$U'^{\dagger}\sigma_zU'=U^{\dagger}\sigma_zU$ is equivalent to
$W^{\dagger}\sigma_zW=\sigma_z$. The stabilizer of $\sigma_z$ under the
adjoint action is precisely the diagonal subgroup $\{V_\psi\}\cong U(1)$,
which proves \eqref{eq:left-orbit}. Since this subgroup acts by multiplication
on the left, the fibres are the orbits of that action, written $U(1)U$, and
the corresponding homogeneous space is $U(1)\backslash SU(2)\cong S^2$.
\end{proof}

\begin{corollary}[No pure ESS in the unreduced $SU(2)$ game]
\label{cor:no-ess-su2} For every game satisfying \eqref{eq:GameDefinition}
and every $\gamma\in[0,\pi/2]$, no pure strategy $U\in SU(2)$ is a strict
symmetric Nash equilibrium, and no pure strategy is an ESS in the sense of
Definition~\ref{def:ESS} applied to the unreduced set $SU(2)$. The conclusion
is unchanged on $SU(2)/\{\pm\mathbb{I}\}$, where the global-sign pair is
identified.
\end{corollary}

\begin{proof}
Fix $U$ and let $U'=V_{\psi}U$ with $V_{\psi}\neq\pm\mathbb{I}$, so that
$U'\neq U$ even after the global-phase identification. Since $V_\psi$
stabilizes $\sigma_z$ under the adjoint action, $n(U')=n(U)$, and by
Proposition~\ref{prop:reduction} the payoff depends on a unitary only through
its Bloch image. The two strategies are consequently payoff-identical against
every opponent at every $\gamma\in[0,\pi/2]$, the case $\gamma=0$ included,
where equality of Bloch images implies the weaker equality of $n_z$ required
there. It follows that $P(U,U)=P(U',U)$, so \eqref{eq:ESS1} fails. The payoff
\eqref{eq:bloch-payoff} is symmetric in its two arguments, so
$P(U,U')=P(U',U)=P(U,U)=P(U',U')$, and the strict inequality in
\eqref{eq:ESS2} fails as well.
\end{proof}

\begin{remark}[Randomized residents in the unreduced game]
\label{rem:mixed-su2} Corollary~\ref{cor:no-ess-su2} extends to randomized
strategies in the sense of Remark~\ref{rem:pure}. For Borel probability
measures $\mu,\nu$ on $SU(2)$, equations \eqref{eq:bloch-payoff} and
\eqref{eq:bilinear-extension} give
\begin{equation}
P(\mu,\nu)=t+\frac{r-t}{2}\Bigl[1+\bar n_\mu^{\top}M_\gamma\,\bar n_\nu\Bigr],
\qquad
\bar n_\mu=\int_{SU(2)}n(U)\,d\mu(U),
\label{eq:mean-bloch-payoff}
\end{equation}
where the mean Bloch vector $\bar n_\mu$ lies in the closed unit ball $B^3$.
A randomized strategy therefore affects payoffs only through its mean Bloch
vector $m=\bar n_\mu$, and at least two distinct measures share every such
mean. If $|m|=1$, then $m$ is an extreme point of $B^3$, so any measure with
mean $m$ is concentrated on the fibre $n^{-1}(m)$; this fibre is a circle, and
$\delta_U$ and $\delta_{V_\psi U}$ with $n(U)=m$ and $V_\psi\neq\pm\mathbb{I}$
are distinct measures with mean $m$. If $|m|<1$, then $m$ lies in the interior
of more than one chord of $S^2$, and two chords give two two-point measures
with mean $m$ and different supports. Hence every randomized resident
$\mu^*$ has a distinct exact payoff duplicate $\nu$, and the argument of
Corollary~\ref{cor:no-ess-su2} shows that $\mu^*$ is not an ESS. The unreduced
$SU(2)$ game thus has no ESS, pure or randomized, for any
$\gamma\in[0,\pi/2]$.
\end{remark}

Payoff equivalence between distinct unitary strategies is a recognized
feature of the EWL scheme rather than a peculiarity of the present game.
Fr\k{a}ckiewicz, Gorczyca-Goraj and Szopa \cite{Frackiewicz2025} treat it as a
basic structural notion, noting that a global phase always produces a
payoff-equivalent unitary and exhibiting further equivalences among phase
parameters. For the agreement-payoff class,
Theorem~\ref{thm:fibre} determines the exact classes within $SU(2)$: they are
the left orbits of a continuous $U(1)$ subgroup. This characterization permits
the equilibrium refinements to be compared before and after payoff reduction.

Theorem~\ref{thm:su2} is a particular case of
Corollary~\ref{cor:no-ess-su2}. Its neutral family \eqref{eq:flat-direction}
is the fibre through $s^{*}$, and the same construction applies to every
other unitary, including those that embed the classical pure strategies. This
is a genuine conclusion for the unreduced normal-form game on $SU(2)$: no pure
unitary type is an ESS. The quotient considered next defines a different,
payoff-reduced game. It preserves the payoff assigned to each matchup of
equivalence classes, while changing the set of distinct mutants, and it need
not preserve strictness or evolutionary stability.

\begin{theorem}[Evolutionary stability in the payoff-reduced game]
\label{thm:reduced} Let the game satisfy \eqref{eq:GameDefinition} and let
the strategy space be the reduced space $S^2\cong U(1)\backslash SU(2)$ with payoff
\eqref{eq:bloch-payoff}. Let $n^{*}=n(s^{*})=\bigl(\tfrac{1}{\sqrt2},
\tfrac{1}{\sqrt2},0\bigr)$ be the image of $s^{*}=(\pi/2,\pi/4,0)$. Then
\begin{enumerate}[label=(\arabic*),nosep]
\item for every $\gamma>0$, $n^{*}$ is a strict symmetric Nash equilibrium
and therefore an ESS;
\item at $\gamma=0$, $n^{*}$ is a non-strict symmetric Nash equilibrium and
is not an ESS.
\end{enumerate}
\end{theorem}

\begin{proof}
Write $g=\sin\gamma$ and $M=M_\gamma$. Since $M$ is symmetric, the reduced
game is doubly symmetric, and by \eqref{eq:bloch-payoff} the Nash gap against
$n^{*}$ is proportional to $n^{*\top}Mn^{*}-a^{\top}Mn^{*}$ for $a\in S^2$.

(1) Let $\gamma>0$. The vector $n^{*}$ is an eigenvector of $M$ with
eigenvalue $g$, since $Mn^{*}=g\,n^{*}$. Hence $n^{*\top}Mn^{*}=g$ and
$a^{\top}Mn^{*}=g\,(a\!\cdot\!n^{*})$. By Cauchy--Schwarz,
$a\!\cdot\!n^{*}\le 1$ for $a\in S^2$, with equality if and only if
$a=n^{*}$. Therefore
\begin{equation}
n^{*\top}Mn^{*}-a^{\top}Mn^{*}=g\bigl(1-a\!\cdot\!n^{*}\bigr)>0
\qquad\text{for all } a\in S^2,\ a\neq n^{*},
\end{equation}
so $n^{*}$ is a strict symmetric Nash equilibrium, and \eqref{eq:ESS1} holds
for every mutant. Hence $n^{*}$ is an ESS.

(2) Let $\gamma=0$. Then $M=\mathrm{diag}(0,0,1)$ and $Mn^{*}=0$, so
$n^{*\top}Mn^{*}-a^{\top}Mn^{*}=0$ for every $a$: the equilibrium is
non-strict and \eqref{eq:ESS1} never applies. For the second condition,
$n^{*\top}Ma-a^{\top}Ma=n^{*}_za_z-a_z^2=-a_z^2$, which is negative for
every mutant with $a_z\neq0$ and vanishes for $a_z=0$. Hence
\eqref{eq:ESS2} fails and $n^{*}$ is not an ESS.
\end{proof}

If one instead forms the maximal payoff quotient separately at $\gamma=0$,
the reduced strategy is $z=n_z\in[-1,1]$ and the resident is $z^*=0$. The
second ESS gap is then $-(r-t)z^2/2$, so the resident is again not an ESS. The
zero-entanglement conclusion is therefore independent of which of the two
natural quotient conventions is used.

\begin{remark}[Randomized and polymorphic mutants; uniform invasion barrier]
\label{rem:mixed} The conclusion of Theorem~\ref{thm:reduced}(1) is not an
artifact of the restriction to pure mutants noted in Remark~\ref{rem:pure}.
Let a mutant population be represented by a Borel probability measure $\mu$ on
$S^2$, read either as a randomized mutant type or as a polymorphic mutant
group, and let $\bar n=\int n\,d\mu\in B^3$. By bilinearity the payoff is given
by \eqref{eq:mean-bloch-payoff}, with $\bar n_\mu=\bar n$. Write $g=\sin\gamma>0$
and $d=n^*-\bar n$. Against the resident $n^*$, the mutants' average Nash gap
is $\tfrac{r-t}{2}A$, where
\begin{equation}
A=g\bigl(1-\bar n\!\cdot\!n^*\bigr)\ge 0 .
\end{equation}
Equality $A=0$ holds only when $\bar n=n^*$; because $n^*$ is an extreme point
of $B^3$ and $\mu$ is supported on $S^2$, this requires $\mu=\delta_{n^*}$.
Every distinct mutant, randomized or polymorphic, therefore satisfies
\eqref{eq:ESS1}.

The invasion barrier is, moreover, uniform. Since $Mn^*=g\,n^*$, the
difference between the two sides of \eqref{eq:invasion-barrier} equals
\begin{equation}
\frac{r-t}{2}\Bigl[(1-\varepsilon)\,d^{\top}Mn^*+\varepsilon\,d^{\top}M\bar n\Bigr]
=\frac{r-t}{2}\Bigl[A-\varepsilon\,d^{\top}Md\Bigr].
\end{equation}
The eigenvalues of $M$ are $1$, $g$ and $-g$, so $d^{\top}Md\le|d|^2$, and
$|d|^2=1-2\bar n\!\cdot\!n^*+|\bar n|^2\le 2(1-\bar n\!\cdot\!n^*)=2A/g$.
Hence the difference is at least $\tfrac{r-t}{2}A(1-2\varepsilon/g)>0$ whenever
\begin{equation}
0<\varepsilon<\bar\varepsilon=\frac{\sin\gamma}{2},
\label{eq:uniform-barrier}
\end{equation}
for every distinct Borel mutant $\mu$. The barrier vanishes as
$\gamma\to0$, consistent with Theorem~\ref{thm:reduced}(2). The same bound
applies in the restricted domain of Theorem~\ref{thm:restricted}, because the
restricted strategies map into $S^2$ and $s^*$ is the only restricted strategy
with Bloch image $n^*$ (Remark~\ref{rem:why-restricted}).

As noted in Remark~\ref{rem:pure}, this global, static stability does not
extend to stability under small changes in the strategy of the whole
population. For $a_\delta=(\cos\delta/\sqrt2,\cos\delta/\sqrt2,\sin\delta)$,
\begin{equation}
P(n^*,a_\delta)-P(a_\delta,a_\delta)
=\frac{r-t}{2}\Bigl[g\cos\delta\,(1-\cos\delta)-\sin^2\delta\Bigr]<0
\qquad(0<|\delta|<\pi/2),
\end{equation}
since $g\cos\delta\le1<1+\cos\delta$. A population playing $a_\delta$, which is
arbitrarily close to $n^*$ in the weak topology as $\delta\to0$, therefore
cannot be invaded by rare $n^*$ players. The ESS property of $n^*$ concerns
invasion of a resident $n^*$ population, and dynamic stability notions for
continuous strategy sets \cite{Eshel,OechsslerRiedel} are not addressed here.
\end{remark}

\begin{remark}[Relation to the restricted domain]
\label{rem:why-restricted} The restricted set \eqref{eq:EWLUnitary} is the
slice $\alpha=0$ with $\phi\in[0,\pi/2]$. Under
\eqref{eq:bloch-coords} it maps onto the closed spherical lune with azimuth
$\chi\in[0,\pi/2]$ and is one-to-one away from the poles. Its image is a proper
subset of the reduced sphere, and Theorem~\ref{thm:restricted} is the
restriction of Theorem~\ref{thm:reduced}(1) to that subset. This explains why
strictness holds in both domains. It does not remove the contrasting result in
the unreduced $SU(2)$ game, where distinct payoff-duplicate unitary types are
retained.
\end{remark}

\section{When is the payoff reduction appropriate?}
\label{sec:when-reduction}

Two separate questions arise. The first is mathematical: for which games does
the payoff factor through the Bloch map, so that the exact payoff-equivalence
classes are the $U(1)$ orbits of Theorem~\ref{thm:fibre}? The second concerns
modelling: when should payoff-equivalent unitaries be identified in an
evolutionary analysis?

\paragraph{The mathematical condition.}
Alice's payoff is the expectation in $|\psi_f\rangle$ of the payoff operator
$\hat\Pi_A=r|S_1S_1'\rangle\langle S_1S_1'|+s|S_1S_2'\rangle\langle S_1S_2'|
+t|S_2S_1'\rangle\langle S_2S_1'|+u|S_2S_2'\rangle\langle S_2S_2'|$. In the
Pauli basis,
\begin{equation}
\hat\Pi_A=\kappa_0\,\mathbb{I}+\kappa_A\,\sigma_z\otimes\mathbb{I}
+\kappa_B\,\mathbb{I}\otimes\sigma_z+\kappa_{AB}\,\sigma_z\otimes\sigma_z,
\label{eq:payoff-operator}
\end{equation}
with
\begin{equation}
\kappa_0=\tfrac14(r+s+t+u),\quad
\kappa_A=\tfrac14(r+s-t-u),\quad
\kappa_B=\tfrac14(r-s+t-u),\quad
\kappa_{AB}=\tfrac14(r-s-t+u).
\end{equation}
The correlation term $\sigma_z\otimes\sigma_z$ commutes with the generator
$\sigma_y\otimes\sigma_y$ of $J$, which is the property used in
Proposition~\ref{prop:reduction}. The single-player terms anticommute with it,
and their Heisenberg-picture images are
\begin{equation}
J(\sigma_z\otimes\mathbb{I})J^{\dagger}=\cos\gamma\,\sigma_z\otimes\mathbb{I}
+\sin\gamma\,\sigma_x\otimes\sigma_y,
\qquad
J(\mathbb{I}\otimes\sigma_z)J^{\dagger}=\cos\gamma\,\mathbb{I}\otimes\sigma_z
+\sin\gamma\,\sigma_y\otimes\sigma_x .
\label{eq:local-terms-heisenberg}
\end{equation}
For $\gamma>0$ these terms bring in $U^{\dagger}\sigma_xU$ and
$U^{\dagger}\sigma_yU$, which vary along the fibres of $U\mapsto n(U)$. The
Bloch-map reduction therefore holds when $\kappa_A=\kappa_B=0$, that is, when
$r=u$ and $s=t$, in agreement with Lemma~\ref{lem:class}. When either
coefficient is nonzero, these terms in general make the payoff depend on more
than the two Bloch vectors, as the explicit example of
Section~\ref{sec:robustness} confirms. The
condition $p^*=1/2$ in \eqref{eq:half-equilibrium-condition} below is
$\kappa_A=0$ alone. It leaves $\kappa_B=(t-s)/2$ free, and this term produces
the profitable deviation found in Section~\ref{sec:robustness}. The same
operator-based reasoning can be applied to other EWL-type protocols: whether a comparable reduction exists is determined by how the
payoff observables transform under conjugation by the entangling gate.
Components commuting with its generator retain the corresponding reduced
structure, whereas noncommuting components may introduce additional dependence
on the local unitaries.

\paragraph{Modelling considerations.}
When the mathematical condition holds, the choice between the unreduced and the
reduced game depends on what the evolutionary model is intended to describe.

The reduced game is appropriate when fitness is determined by the payoff of the
game alone. Members of an exact payoff-equivalence class then have equal
fitness in every population state, so selection cannot discriminate between
them. For finitely many types under the replicator dynamics
$\dot x_i=x_i\bigl(P(i,x)-P(x,x)\bigr)$, the total frequency $x_K$ of a class
$K$ obeys $\dot x_K=x_K\bigl(P(K,x)-P(x,x)\bigr)$, which is the replicator
dynamics of the reduced game, while the relative frequencies $x_i/x_K$ within
$K$ remain constant \cite{HofbauerSigmund}. The reduced game thus describes the
whole selective part of the dynamics, and the distinction between unitaries
within a class is selectively neutral. Questions of the form ``which
payoff-relevant behaviour is favoured by selection?'' are posed most naturally
in the reduced game.

The unreduced game is appropriate when members of a class must be kept apart
for reasons outside the payoff function. Examples are strategies that differ in
implementation cost, circuit complexity or robustness to noise; a population
that also plays other games, or is subject to noise channels, that separate
members of a class; and situations in which the full outcome distribution, and
not only the payoff, is recorded. Payoff-equivalent unitaries in general
produce different full distributions $(P_{CC},P_{CD},P_{DC},P_{DD})$. For
example, when the members of the fibre \eqref{eq:perturbation-fibre} play
against $s^*$ at $\gamma=\pi/4$, the probability $P_{CD}+P_{DD}$ that Bob's
outcome is $D$ varies with $\varphi$ according to \eqref{eq:qphi}, whereas
the agreement-game payoff is the same for every member. Outside the agreement-payoff
class the $U(1)\backslash SU(2)$ payoff quotient of
Theorem~\ref{thm:fibre} no longer applies, and the unreduced analysis is
required (Section~\ref{sec:robustness}).

The restricted two-parameter domain is a third option. It is appropriate when a
physical constraint limits the available operations, although such
restrictions should be justified independently \cite{BenjaminHayden}.
In the present game the restriction does not affect the conclusion for the
resident $s^*$: the fibre over $n^*$ meets the restricted domain only at $s^*$,
so the restricted game contains no payoff duplicate of $s^*$
(Remark~\ref{rem:why-restricted}).

\paragraph{Reconciling the two formulations.}
The unreduced and reduced answers are related by a set-valued notion of
stability. In the unreduced game the resident $s^*$ is not an ESS, but its
whole payoff-equivalence class is stable as a set, in a sense closely related
to the evolutionarily stable sets of Thomas \cite{Thomas}.

\begin{proposition}[Neutral stability and set-valued stability of the fibre]
\label{prop:ess-set} Let the game satisfy \eqref{eq:GameDefinition}, let
$\gamma\in(0,\pi/2]$, and let
$F=\{U\in SU(2):n(U)=n^*\}=\{V_\psi U^*:\psi\in[0,2\pi)\}$ be the fibre through
$U^*=U(\pi/2,\pi/4,0)$. In the unreduced game on $SU(2)$, with mutants
described as in Remark~\ref{rem:pure}:
\begin{enumerate}[label=(\roman*),nosep]
\item every $U\in F$ is a neutrally stable strategy;
\item for every $U\in F$ and every Borel probability measure $\mu$ on $SU(2)$
with $\mu(F)<1$, $P(U,U)>P(\mu,U)$, and the resident $U$ resists $\mu$ in the
sense of \eqref{eq:invasion-barrier} for all
$0<\varepsilon<\sin\gamma/2$;
\item for every $U\in F$ and every $\mu$ with $\mu(F)=1$, all four payoffs
$P(U,U)$, $P(\mu,U)$, $P(U,\mu)$ and $P(\mu,\mu)$ are equal.
\end{enumerate}
Thus, as a set, $F$ resists every mutant population not concentrated on it,
with a uniform invasion barrier, while mutants within $F$ are selectively
neutral. At $\gamma=0$, $U^*$ is not neutrally stable.
\end{proposition}

\begin{proof}
By \eqref{eq:mean-bloch-payoff}, every payoff depends on a measure only through
its mean Bloch vector, and $\bar n_{\delta_U}=n^*$ for $U\in F$. As in
Remark~\ref{rem:mixed}, $\bar n_\mu=n^*$ if and only if $\mu$ is concentrated
on $n^{-1}(n^*)=F$. Part (iii) follows because all four payoffs equal
$t+\tfrac{r-t}{2}(1+n^{*\top}M_\gamma n^*)$. For part (ii), the Nash gap and
the invasion-barrier difference depend only on $\bar n_\mu$ and coincide with
those of Remark~\ref{rem:mixed}, which are positive when $\bar n_\mu\neq n^*$
and $0<\varepsilon<\sin\gamma/2$. Parts (ii) and (iii) together give (i). At
$\gamma=0$, any pure mutant with $n_z\neq0$ has zero Nash gap and second ESS gap
$-\tfrac{r-t}{2}n_z^2<0$ by \eqref{eq:ESSSecondQ_SU2}, so neutral stability
fails.
\end{proof}

Proposition~\ref{prop:ess-set} shows that the statement ``$n^*$ is an ESS of
the reduced game'' corresponds to the set-valued stability of the fibre $F$ in
the unreduced game. The two formulations agree on the relevant payoff
comparisons and differ in whether a single unitary or a payoff-equivalence
class is regarded as the unit of stability.

\section{Scope: games outside the agreement-payoff class}

\label{sec:robustness} Lemma~\ref{lem:class} shows that
\eqref{eq:GameDefinition} is exactly the condition for the outcome payoff to
be an affine function of the total agreement probability. The
agreement-observable reduction in Proposition~\ref{prop:reduction} follows
from this property. Because the condition imposes two equations on the four
payoff entries, it is natural to ask whether the equilibrium result survives
perturbations that preserve the underlying classical mixed equilibrium.

For the general symmetric game \eqref{eq:Matrix}, an interior symmetric mixed
equilibrium satisfies
\begin{equation}
p^*=\frac{u-s}{r-s-t+u},
\label{eq:general-mixed-equilibrium}
\end{equation}
when the denominator is nonzero. Hence $p^*=1/2$ is equivalent to
\begin{equation}
r+s=t+u,
\qquad\text{or equivalently}\qquad u-s=r-t.
\label{eq:half-equilibrium-condition}
\end{equation}
Consider the arbitrarily small perturbation
\begin{equation}
(r,s,t,u)=(1,\varepsilon,0,1+\varepsilon),
\qquad 0<\varepsilon<\frac12.
\label{eq:perturbed-family}
\end{equation}
It remains a coordination game, satisfies
\eqref{eq:half-equilibrium-condition}, and retains the classical
interior equilibrium $p^*=1/2$, but it lies outside the agreement-payoff class.

The loss of the quantum equilibrium admits a closed-form demonstration. Set
$\gamma=\pi/4$ and consider the former neutral fibre through $s^*$,
\begin{equation}
m_\varphi=\left(\frac{\pi}{2},\varphi,
\varphi-\frac{\pi}{4}\right), \qquad 0\le\varphi<2\pi,
\label{eq:perturbation-fibre}
\end{equation}
where the phase coordinates are understood modulo $2\pi$. All these unitaries
have the same Bloch image as $s^*$. For the payoffs
\eqref{eq:perturbed-family}, however, Alice's payoff can be written as
\begin{equation}
P_A=(P_{CC}+P_{DD})+\varepsilon(P_{CD}+P_{DD}).
\end{equation}
The first term is constant along \eqref{eq:perturbation-fibre}, while a direct
evaluation of the EWL circuit gives
\begin{equation}
q(\varphi):=P_{CD}+P_{DD}
=\frac12+\frac{\sqrt2}{8}\sin\left(2\varphi+\frac{\pi}{4}\right)
-\frac14\cos\left(2\varphi+\frac{\pi}{4}\right).
\label{eq:qphi}
\end{equation}
At the resident, $q(\pi/4)=(5+\sqrt2)/8$. The oscillatory part of
\eqref{eq:qphi} has amplitude
\begin{equation}
\sqrt{\left(\frac{\sqrt2}{8}\right)^2+\left(\frac14\right)^2}
=\frac{\sqrt6}{8},
\end{equation}
and therefore
\begin{equation}
\max_\varphi q(\varphi)=\frac{4+\sqrt6}{8}.
\end{equation}
Consequently
\begin{equation}
\min_\varphi\bigl[P_A(s^*,s^*)-P_A(m_\varphi,s^*)\bigr]
=-\frac{\varepsilon}{8}\bigl(\sqrt6-\sqrt2-1\bigr)<0.
\label{eq:profitable-perturbation}
\end{equation}
Hence for every $\varepsilon$ admitted by \eqref{eq:perturbed-family} there
exists a profitable deviation from $s^*$. The classical equilibrium remains at
$p^*=1/2$, while $s^*$ fails to be a Nash equilibrium of the quantized game at
$\gamma=\pi/4$.

This family establishes that the reduced-space stability theorem is not robust
against arbitrary payoff perturbations outside the agreement-payoff class. It
does not establish that every game outside the class lacks a related quantum
equilibrium. Where the payoff no longer reduces to the agreement observable,
any such equilibrium must be located and analysed independently.

\section{Summary of the stability results}
\label{sec:summary}

Table~\ref{tab:stability-summary} collects the conclusions proved above. The
table separates the strategic domain from the value of the entanglement
parameter, since both determine whether payoff-neutral mutants are retained as
distinct evolutionary types.

\begin{table}[htbp]
\centering
\caption{Stability of the equilibrium corresponding to the classical mixed
strategy $p^*=1/2$. In the unreduced $SU(2)$ game the absence of an ESS
holds for every $\gamma\in[0,\pi/2]$ by Corollary~\ref{cor:no-ess-su2} and
Remark~\ref{rem:mixed-su2}, and is not confined to nonzero entanglement. The
set-valued stability for $\gamma>0$ is Proposition~\ref{prop:ess-set}. Where
an ESS is stated for $\gamma>0$, the invasion barrier $\sin\gamma/2$ is
uniform over all Borel mutants (Remark~\ref{rem:mixed}).}
\label{tab:stability-summary}
\begin{tabular}{@{}>{\raggedright\arraybackslash}p{0.30\textwidth}>{\raggedright\arraybackslash}p{0.31\textwidth}>{\raggedright\arraybackslash}p{0.31\textwidth}@{}}
\toprule
Strategic formulation & $\gamma=0$ & $\gamma>0$ \\ \midrule
Classical interior mixed equilibrium $p^*=1/2$ & Nash equilibrium; not an ESS
& Not applicable \\[2pt]
Restricted EWL strategy $s^*=(\pi/2,\pi/4)$ & Non-strict Nash equilibrium;
not an ESS & Strict Nash equilibrium; ESS \\[2pt]
Unreduced $SU(2)$ strategy $s^*=(\pi/2,\pi/4,0)$ & Non-strict Nash
equilibrium; not neutrally stable; not an ESS; the game has no pure or
randomized ESS & Non-strict Nash equilibrium; neutrally stable but not an ESS;
its payoff class is stable as a set; the game has no pure or
randomized ESS \\[2pt]
Payoff-reduced strategy $n^*\in U(1)\backslash SU(2)\cong S^2$ & Non-strict
Nash equilibrium; not an ESS & Strict Nash equilibrium; ESS \\ \bottomrule
\end{tabular}
\end{table}

\section{Discussion and conclusion}

\label{sec:discussion} For the classical agreement-payoff coordination game,
the interior mixed equilibrium $p^*=1/2$ is a symmetric Nash equilibrium but
not an ESS. In the restricted two-parameter EWL domain, the associated strategy
$s^*=(\pi/2,\pi/4)$ becomes strict, and hence evolutionarily stable, for every
$\gamma>0$. The result follows directly from the positive Nash gap in
\eqref{eq:NEQG}.

The full $SU(2)$ analysis gives two answers according to the chosen strategy
identity. If every unitary is retained as a distinct pure type, $s^*$ is a
non-strict Nash equilibrium and not an ESS. More generally, no pure unitary is
an ESS, because each has a continuum of exact payoff duplicates, and the same
holds for randomized strategies (Remark~\ref{rem:mixed-su2}). The unitaries in a
common class are in general distinct physical operations, and their equivalence
is a property of the payoff function of the present game rather than of the
operations themselves.

The payoff-reduced formulation identifies these exact duplicates. The
agreement payoff is an affine function of the expectation of
$\sigma_z\otimes\sigma_z$, which commutes with the entangler. Consequently a
local unitary enters the payoff only through
$U^{\dagger}\sigma_zU=n(U)\cdot\boldsymbol{\sigma}$. For $\gamma>0$, the fibres of this map are exactly
the payoff-equivalence classes, and the quotient is
$U(1)\backslash SU(2)\cong S^2$. On this quotient, $n^*$ is a strict symmetric
Nash equilibrium for every $\gamma>0$. The strictness also holds against
randomized and polymorphic mutants, with a uniform invasion barrier, as shown
in Remark~\ref{rem:mixed}.

The unreduced and quotient results are not contradictory. Exact duplication does
not alter payoff values, but ESS is sensitive to whether payoff-identical types
are counted as distinct mutants. Proposition~\ref{prop:ess-set} makes the
correspondence precise: for $\gamma>0$ the resident is neutrally stable in the
unreduced game, and its payoff-equivalence class is stable as a set, which is
the set-valued counterpart of the ESS found in the reduced game.
As discussed in Section~\ref{sec:when-reduction}, the quotient is appropriate
when fitness is determined by the game payoff alone, since the replicator
dynamics then acts on payoff-equivalence classes and is neutral within them.
The unreduced game is appropriate when distinct unitary implementations must be
retained as distinct evolutionary types for reasons outside the payoff
function. Claims about entanglement-induced stability should therefore state
which convention is used.

Three limitations should be noted. First, the ESS property established here is
the static, global property of Definition~\ref{def:ESS}, with a uniform
invasion barrier $\sin\gamma/2$ against every Borel mutant population
(Remark~\ref{rem:mixed}). It does not imply stability with respect to small
changes in the strategy of the whole population, and Remark~\ref{rem:mixed}
exhibits nearby populations that rare $n^*$ players cannot invade. Second, the
result concerns the evolutionary status of the interior equilibrium and not
equilibrium selection. In the reduced game the images $\pm e_z$ of the
classical pure coordination strategies remain strict symmetric Nash equilibria
for every $\gamma$, with payoff $r$, which exceeds the payoff
$t+\tfrac{r-t}{2}(1+\sin\gamma)$ of $n^*$ for $\gamma<\pi/2$ and equals it at
$\gamma=\pi/2$. Third, the conclusion is confined to the
agreement-payoff class $r=u$ and $s=t$. Outside this class, single-player terms
in the payoff operator make the payoff depend on more than the Bloch vectors,
and the reduction no longer follows. The perturbation
\eqref{eq:perturbed-family} preserves the classical equilibrium $p^*=1/2$ but
produces a profitable quantum deviation at $\gamma=\pi/4$. Thus neither the
location nor the stability of the quantum equilibrium extends to arbitrary
nearby payoff matrices. A corresponding classification for general symmetric
$2\times2$ payoffs lies beyond the scope of the present analysis.

These results bear on quantum game theory in three ways. They extend the
Benjamin--Hayden strategy-space dependence from Nash equilibrium to
evolutionary stability, and identify payoff equivalence as the specific
mechanism by which ESS can be lost when the strategy space is enlarged. They
provide a geometric criterion, based on the commutation of the measured
observable with the entangling gate, that identifies the payoff-equivalence
classes in the present game and a natural payoff-reduced strategy space, and
that indicates how comparable reductions may be sought in other EWL-type
protocols. Whether that quotient or the unreduced strategy space is the
appropriate evolutionary model remains a modelling choice, as discussed in
Section~\ref{sec:when-reduction}. Finally,
they indicate which conclusions are robust in applications that use quantum
games to model populations of adaptive agents, such as entanglement-assisted
coordination: the stability of the quantum strategy holds for a class of
payoff-equivalent strategies rather than for a single unitary, and it depends
on an exact symmetry of the payoffs that small perturbations destroy.

\section*{Statements and Declarations}

\noindent\textbf{Funding.}
\textit{The funding support from the Australian Research Council (FL240100217) is gratefully
acknowledged.
}

\medskip
\noindent\textbf{Competing interests.}
The authors declare no competing interests.

\medskip
\noindent\textbf{Data availability.}
No datasets were generated or analysed during the present study. All results
are analytic and are derived in the text and appendices.

\medskip
\noindent\textbf{Code availability.}
The numerical verification scripts used to check the analytic formulas are
included with the manuscript source files.

\medskip
\noindent\textbf{Author contributions.}
Both authors contributed to the conception, mathematical analysis, and
preparation of the manuscript. Both authors read and approved the final
manuscript.

\appendix

\section{Derivation of the compact two-parameter payoff}

\label{app:derivations} For completeness we outline the algebra leading to
the compact form \eqref{eq:PayoffsQG}. In the EWL two-parameter
convention the measurement probabilities are given by \eqref{eq:PCC-2p}--%
\eqref{eq:PDC-2p}. Under the agreement-payoff assumptions $s=t$ and $u=r$, the
common payoff is 
\begin{equation}
P(s_A,s_B)=t+(r-t)(P_{CC}+P_{DD}).
\end{equation}
Now 
\begin{align}
P_{CC} &=c_A^2c_B^2\bigl(\cos^2\Phi+\cos^2\gamma\,\sin^2\Phi\bigr) %
=c_A^2c_B^2\bigl(1-\sin^2\gamma\,\sin^2\Phi\bigr), \\
P_{DD} &=\bigl(s_As_B+\sin\gamma\,\sin\Phi\,c_Ac_B\bigr)^2
=s_A^2s_B^2+2\sin\gamma\,\sin\Phi\,c_Ac_Bs_As_B+\sin^2\gamma\,\sin^2\Phi%
\,c_A^2c_B^2.
\end{align}
Adding these gives \eqref{eq:pccpdd-simplified}: 
\begin{equation}
P_{CC}+P_{DD}=c_A^2c_B^2+s_A^2s_B^2+2\sin\gamma\,\sin\Phi\,c_Ac_Bs_As_B.
\end{equation}
Applying the elementary identities 
\begin{equation}
2c_As_A=\sin\theta_A, \qquad 2c_Bs_B=\sin\theta_B, \qquad
c_A^2c_B^2+s_A^2s_B^2=\frac{1+\cos\theta_A\cos\theta_B}{2}
\end{equation}
yields the compact expression \eqref{eq:PayoffsQG}.

\section{Derivation of the full pure \texorpdfstring{$SU(2)$}{SU(2)} formulas%
}

\label{app:su2-derivation} In the full pure $SU(2)$ extension %
\eqref{eq:qstrategy_su2}, write
\begin{equation}
\Phi=\phi_A+\phi_B, \qquad A=\alpha_A+\alpha_B.
\end{equation}
A direct calculation of
\begin{equation}
|\psi_f\rangle=J^{\dagger}(U_A\otimes U_B)J|S_1S_1^{\prime }\rangle
\end{equation}
with the same entangler \eqref{eq:EWLentangler} gives the two amplitudes
relevant under $s=t$ and $u=r$:
\begin{align}
\langle S_1S_1^{\prime}|\psi_f\rangle
&=c_Ac_B\cos\Phi-\sin\gamma\,d_Ad_B\sin A
+i\cos\gamma\,c_Ac_B\sin\Phi,
\\
\langle S_2S_2^{\prime}|\psi_f\rangle
&=\sin\gamma\,c_Ac_B\sin\Phi+d_Ad_B\cos A
+i\cos\gamma\,d_Ad_B\sin A.
\end{align}
Taking squared moduli gives \eqref{eq:PCC-su2} and \eqref{eq:PDD-su2}.
Because $s=t$ and $u=r$, only the sum $P_{CC}+P_{DD}$ enters the common
payoff. To obtain \eqref{eq:Psstar_s_su2}, substitute 
\begin{equation}
s^*=\left(\frac{\pi}{2},\frac{\pi}{4},0\right), \qquad
s=(\theta,\phi,\alpha),
\end{equation}
so that 
\begin{equation}
c_A=d_A=\frac{1}{\sqrt 2}, \qquad \phi_A=\frac{\pi}{4}, \qquad \alpha_A=0,
\qquad c_B=\cos\frac{\theta}{2}, \qquad d_B=\sin\frac{\theta}{2}.
\end{equation}
Substitution into \eqref{eq:PCC-su2} and \eqref{eq:PDD-su2}, followed by
elementary trigonometric simplification, yields 
\begin{equation}
P_{CC}+P_{DD}=\frac12\left[1+\sin\gamma\,\sin\theta\,\sin\left(\phi+\frac{\pi%
}{4}-\alpha\right)\right],
\end{equation}
which gives \eqref{eq:Psstar_s_su2}. Setting $s_A=s_B=s$ yields 
\begin{equation}
P_{CC}+P_{DD}=\frac12\left[1+\cos^2\theta+\sin\gamma\,\sin^2\theta\,\sin%
\bigl(2(\phi-\alpha)\bigr)\right],
\end{equation}
which gives \eqref{eq:Pss_su2}. Subtracting the two expressions leads
immediately to \eqref{eq:ESSSecondQ_SU2}, and comparing $P(s^*,s^*)$ with $%
P(s,s^*)$ gives \eqref{eq:NEdiff_su2}.

\section{The Bloch-vector reduction}

\label{app:bloch} This appendix supplies the two computations used in
Section~\ref{sec:reduction}.

\paragraph{The Bloch image.}
For $U=U(\theta,\phi,\alpha)$ as in \eqref{eq:qstrategy_su2}, a direct
evaluation of the vector in
$U^{\dagger}\sigma_zU=n(U)\cdot\boldsymbol{\sigma}$, using
$n_k=\tfrac12\tr\!\left(U^{\dagger}\sigma_zU\,\sigma_k\right)$, gives
\begin{equation}
n_x=\sin\theta\cos(\alpha-\phi), \qquad
n_y=-\sin\theta\sin(\alpha-\phi), \qquad
n_z=\cos\theta ,
\end{equation}
whence $n_x^2+n_y^2+n_z^2=1$ and, on setting $\chi=\phi-\alpha$,
\begin{equation}
n=(\sin\theta\cos\chi,\ \sin\theta\sin\chi,\ \cos\theta),
\end{equation}
which is \eqref{eq:bloch-coords}. The parameters $\phi$ and $\alpha$ enter
only through $\chi$, so the fibre over a given $n$ is the one-parameter set
obtained by translating $\phi$ and $\alpha$ together, in agreement with
Theorem~\ref{thm:fibre}. Setting $\alpha=0$ recovers the restricted
two-parameter set, for which $\chi=\phi$.

\paragraph{The correlation function.}
With $|\psi_i\rangle=\cos(\gamma/2)|S_1S_1'\rangle+i\sin(\gamma/2)
|S_2S_2'\rangle$ and $a,b\in S^2$,
\begin{align}
\langle\psi_i|\,\bigl(a\!\cdot\!\boldsymbol{\sigma}\bigr)
\otimes\bigl(b\!\cdot\!\boldsymbol{\sigma}\bigr)\,|\psi_i\rangle
&=\cos^2\tfrac{\gamma}{2}\,a_zb_z+\sin^2\tfrac{\gamma}{2}\,a_zb_z
+2\sin\tfrac{\gamma}{2}\cos\tfrac{\gamma}{2}\,(a_xb_y+a_yb_x) \notag \\
&=a_zb_z+\sin\gamma\,(a_xb_y+a_yb_x),
\end{align}
the cross terms in $a_xb_x$ and $a_yb_y$ cancelling against one another
because of the relative factor $i$ in $|\psi_i\rangle$. This is the quadratic
form $a^{\top}M_\gamma b$ of \eqref{eq:bloch-payoff}. At $\gamma=0$ only the
$a_zb_z$ term survives, which is the classical correlation; at $\gamma=\pi/2$
the matrix $M_\gamma$ has spectrum $\{-1,1,1\}$.

\paragraph{Consistency check.}
Substituting $\alpha_A=\alpha_B=0$ into \eqref{eq:bloch-payoff} gives
\begin{equation}
P=t+\frac{r-t}{2}\Bigl[1+\cos\theta_A\cos\theta_B+\sin\gamma\,
\sin\theta_A\sin\theta_B\sin(\phi_A+\phi_B)\Bigr],
\end{equation}
since $a_xb_y+a_yb_x=\sin\theta_A\sin\theta_B\sin(\phi_A+\phi_B)$; this is
\eqref{eq:PayoffsQG}. Substituting $s^{*}=(\pi/2,\pi/4,0)$ and
$s=(\theta,\phi,\alpha)$ gives $a=(\tfrac1{\sqrt2},\tfrac1{\sqrt2},0)$ and
$a^{\top}M_\gamma b=\sin\gamma\,\sin\theta\,\sin(\phi+\pi/4-\alpha)$,
reproducing \eqref{eq:Psstar_s_su2}; and setting $s_A=s_B=s$ gives
$b^{\top}M_\gamma b=\cos^2\theta+\sin\gamma\sin^2\theta\sin(2(\phi-\alpha))$,
reproducing \eqref{eq:Pss_su2}.

\end{document}